\documentclass{IEEEtran}
\usepackage{cite}
\usepackage{amsmath,amssymb,amsfonts}
\usepackage{algorithmic}
\usepackage{graphicx}
\usepackage{textcomp}
\usepackage{xcolor}
\usepackage{mathrsfs}
\usepackage{psfrag}
\usepackage{cancel}
\usepackage{verbatim}
\usepackage{tikz}\usetikzlibrary{arrows.meta,decorations.pathmorphing,angles,quotes,calc}

\makeatletter
\let\NAT@parse\undefined
\makeatother
\usepackage{hyperref}

\def\mc{\mathcal}
\def\mb{\mathbb}

\newtheorem{definition}{Definition}
\newtheorem{theorem}{Theorem}

\newtheorem{lemma}{Lemma}
\newtheorem{remark}{Remark}

\newtheorem{assumption}{Assumption}

\begin{document}
\title{Absolute Stability of Nonlinear Negative Imaginary Systems with Application to Potential Energy Shaping}

\author{Kanghong Shi, \qquad Ian R. Manchester
\thanks{This work was supported by the Australian Research Council (DP230101014).}
\thanks{K. Shi and I. R. Manchester are with the Australian Centre for Robotics, School of Aerospace, Mechanical and Mechatronic Engineering, The University of Sydney, Sydney, NSW 2006, Australia.
 {\tt kanghong.shi@outlook.com}, {\tt ian.manchester@sydney.edu.au}}
}

\maketitle

\thispagestyle{plain}
\pagestyle{plain}

\begin{abstract}
	This paper establishes absolute stability conditions for nonlinear negative imaginary (NI) systems interconnected with static nonlinear feedback.
We first show that the NI property is preserved when the feedback nonlinearity can be expressed as the gradient of a continuously differentiable function, and the composite storage of the resulting system remains positive definite. This condition provides a direct connection between nonlinear static feedback and storage-function shaping along the measured output channels.
Building on this result, conditions are derived for absolute stability of the closed-loop system under mild assumptions.
The linear specialization of the results strictly generalizes prior absolute stability results for linear NI systems, allowing coupled nonlinearities not covered by existing slope-restricted or sector-bounded frameworks.
Finally, the proposed theory is illustrated through a linear example highlighting this generalization and a nonlinear example that shows the utility of the proposed results in potential energy shaping.
\end{abstract}


\section{Introduction}

Negative imaginary (NI) systems theory \cite{lanzon2008stability,petersen2010feedback} provides a robust control framework for mechanical systems with colocated force actuators and position sensors. Roughly speaking, a system is said to be NI if it is stable and the imaginary part of its frequency response is nonpositive for all positive frequencies. While positive real (PR) systems theory applies to systems with relative degree up to one, NI systems can have relative degree zero, one, or two. Under suitable conditions, an NI system can be asymptotically stabilized by applying a strictly NI controller in positive feedback. NI systems theory has found applications in several areas including nanopositioning \cite{bhikkaji2009fast,das2015multivariable} and the control of lightly damped flexible structures \cite{bhikkaji2011negative,rahman2015design}.

Since most engineering systems are inherently nonlinear, the absolute stability problem for NI systems was investigated in \cite{dey2016absolute}. The absolute stability problem studies the conditions under which the interconnection of a dynamical system and a static nonlinearity, known as a Lur'e system \cite{khalil2002nonlinear}, is stable. Classical results such as the circle criterion \cite{zames1966input} and the Popov criterion \cite{popov1961absolute} provide sector-bound conditions for the absolute stability of linear systems. In \cite{dey2016absolute}, absolute stability conditions were derived for linear plants with the strongly strictly negative imaginary (SSNI) property. The results were later extended in \cite{carrasco2017comment}. These works provide a useful framework for analyzing simple nonlinearities such as saturation and dead zones. However, both \cite{dey2016absolute} and \cite{carrasco2017comment} considered only diagonal nonlinearities of the form
\(\phi(y)=\begin{bmatrix}\phi_1(y_1),\ldots,\phi_p(y_p)\end{bmatrix}^\top,\)
which significantly limits their applicability.

The NI systems framework was later extended to nonlinear systems in \cite{ghallab2018extending} through the notion of counterclockwise input-output dynamics \cite{angeli2006systems}. A nonlinear NI system can be characterized as a dissipative system with a supply rate given by the inner product of the input and the time derivative of the output. While stability results have been established for interconnected nonlinear NI systems \cite{ghallab2018extending,shi2021robust}, a general and practical solution to the absolute stability problem has remained elusive. Even in the nonlinear systems scenario, it is still of theoretical and practical significance to investigate the effect of a memoryless, possibly cross-coupled nonlinear feedback on a nonlinear NI system.

In this paper, we investigate the Lur’e interconnection of a nonlinear NI system with a static nonlinear feedback and derive conditions under which the NI property is preserved, as well as conditions ensuring absolute stability. We show that a nonlinear static feedback can effectively modify the system’s storage function along its measured output channels. From this perspective, we view the nonlinearity as the gradient of a scalar function, through which the key conditions are obtained. For physical systems such as mechanical plants, this corresponds to shaping the potential energy of the system. Based on the same principle, we also establish absolute stability conditions when the plant has certain strictness. The corresponding linear specialization of the result is also derived and shown to strictly generalize the results in \cite{dey2016absolute} and \cite{carrasco2017comment}. This generalization is illustrated with a linear numerical example, followed by a nonlinear two-pendulum example where the system's potential energy is shaped using the proposed results to achieve certain control objectives.

The contribution of this work is not only in providing absolute stability conditions for nonlinear NI systems, but also in introducing a constructive method to modify a nonlinear NI system's storage function. Given any desired modification in the storage along the measured outputs, our results show that it can be realized by applying the gradient of that modification as a static nonlinear feedback. This has direct applications to control tasks such as compliant control, where the goal is to adjust stiffness between joints and the environment or between multiple joints. In comparison with prior results, the improved generality and applicability of the proposed results are also evident from their physical interpretation. The results in \cite{dey2016absolute} and \cite{carrasco2017comment} considered only diagonal nonlinearities, which do not allow modification of stiffness between measured coordinates. Moreover, \cite{dey2016absolute} permits only softening each coordinate, while \cite{carrasco2017comment} allows either softening or stiffening all coordinates, but not independently. In contrast, our results can deal with cross-coupling and sign-indefinite storage modifications. Recent nonlinear NI results with nonlinear feedthrough \cite{chen2024nonlinear,chen2024stability} generalized the interconnection structures in \cite{shi2021robust,shi2023output}, but \cite{chen2024nonlinear} is limited to single-input single-output systems, and \cite{chen2024stability} assumes uncoupled, diagonal nonlinearities. Therefore, these existing results are special cases of the results presented in the present paper.

The remainder of this paper is organized as follows. Section~\ref{sec:pre} reviews preliminaries on nonlinear NI systems. Section~\ref{sec:preservation} presents NI preservation under suitable static nonlinear feedback and its physical interpretation. Section~\ref{sec:absolute stability} develops absolute stability conditions for nonlinear NI systems, specializes them to the linear case, and compares them with existing linear results. Section~\ref{sec:examples} illustrates the proposed theory with two examples: a linear example demonstrating generalization over previous linear results, and a nonlinear pendulum example demonstrating the utility of the proposed results in potential energy shaping. Section~\ref{sec:conclusion} concludes the paper.

Notation: The notation in this paper is standard. $\mathbb R$ denotes the set of real numbers. $\mb N_+$ denotes the set of positive integers. $\mathbb R^{m\times n}$ denotes the space of real matrices of dimension $m\times n$. $A^\top$ and $A^{-1}$ denote the transpose and the inverse of a matrix $A$, respectively. $\|x\|$ denotes the standard Euclidean norm of a vector $x$. For a real symmetric or complex Hermitian matrix $P$, $P>0\ (P\geq 0)$ denotes the positive (semi-)definiteness of a matrix $P$ and $P<0\ (P\leq 0)$ denotes the negative (semi-)definiteness of a matrix $P$. A function $V: \mb R^n \to \mb R$ is called positive definite if $V(0)=0$ and $V(x)>0$ for all $x\neq 0$. A function $V: \mb R^n \to \mb R$ is said to be of class $C^1$ if it is continuously differentiable. $\nabla$ denotes the gradient operator. $\mathrm{diag}(a_1,\cdots,a_n)$ denotes a diagonal matrix with the elements $a_1,\cdots,a_n$ on its diagonal.

\section{Preliminaries}\label{sec:pre}
Consider the following general nonlinear system:
\begin{subequations}\label{eq:general nonlinear system}
	\begin{align}
    \dot x(t)=&\ f(x(t),u(t)),\label{eq:state equation of nonlinear NI}\\
    y(t)=&\ h(x(t)),
    \label{eq:output equation of nonlinear NI}
\end{align}
\end{subequations}
where $x(t)\in \mathbb R^{n}$ is the state, $u(t)\in \mathbb R^p$ is a locally integrable input, and $y(t)\in \mathbb R^p$ ($p\leq n$) is the output, $f:\mathbb R^n\times \mathbb R^p \to \mathbb R^n$ is a Lipschitz continuous function and $h:\mathbb R^n \to \mathbb R^p$ is a class $C^1$ function. Without loss of generality, suppose $f(0,0)=0$ and $h(0)=0$.
\begin{definition}[nonlinear NI systems]\label{def:nonlinear NI}\cite{ghallab2018extending,shi2023output}
The system (\ref{eq:general nonlinear system}) is said to be a negative imaginary (NI) system if there exists a positive definite storage function $V:\mathbb R^n\to \mathbb R$ of class $C^1$ such that for any locally integrable input $u$ and any solution $x$ to (\ref{eq:general nonlinear system}),
\begin{equation}\label{eq:NI inequality}
    \dot V(x(t))\leq u(t)^{\top}\dot y(t),
\end{equation}
for all $t\geq 0$.
\end{definition}

\begin{definition}[nonlinear OSNI systems]\label{def:nonlinear OSNI}\cite{shi2021robust,shi2023output} (also see e.g., \cite{bhowmickoutput} for linear OSNI systems)
The system (\ref{eq:general nonlinear system}) is said to be an output strictly negative imaginary (OSNI) system if there exists a positive definite storage function $V:\mathbb R^n\to \mathbb R$ of class $C^1$ and a scalar $\epsilon>0$ such that for any locally integrable input $u$ and any solution $x$ to (\ref{eq:general nonlinear system}),
\begin{equation}\label{eq:OSNI inequality}
    \dot V(x(t))\leq u(t)^{\top}\dot y(t)-\epsilon\|\dot y\|^2,
\end{equation}
for all $t\geq 0$.
\end{definition}

Consider a Hamiltonian system of the following form
\begin{subequations}\label{eq:Hamiltonian}
	\begin{align}
	\dot x =&\ [J(x)-R(x)]\left(\nabla H(x)-\nabla C(x)^{\top} u\right),\label{eq:pH state}\\
		y =&\ C(x),\label{eq:pH output}
	\end{align}
\end{subequations}
where $x\in \mb R^n$, $u,y\in \mb R^p$ are the state, input and output of the system, respectively. Here, $J:\mb R^n\to \mb R^{n\times n}$ is skew-symmetric; i.e., $J(x) = -J(x)^{\top}$. Also, $R:\mb R^n\to \mb R^{n\times n}$ is symmetric; i.e., $R(x) = R(x)^{\top}$. Here, $H:\mb R^n\to \mb R$ and $C:\mb R^n \to \mb R^p$ are class $C^1$ functions with $H(0)=0$ and $C(0)=0$.

\begin{lemma}[NI input-output Hamiltonian system]\cite{van2011positive,van2016interconnections}
	Consider the system given by (\ref{eq:Hamiltonian}). If $H(x)$ is a positive definite function and $R(x)\geq 0$ for all $x\in \mb R^n$, then the system (\ref{eq:Hamiltonian}) is a nonlinear NI system with storage function $H(x)$.
\end{lemma}

\section{NI preservation under static nonlinearities}\label{sec:preservation}
We consider a Lur'e structure of an NI system $\mc H$ as shown in Fig.~\ref{fig:interconnection}, where its output is fed back to its input via a static nonlinearity. We investigate the conditions on the nonlinearity under which the resulting system is still NI. Also, we provide the physical interpretation of such a static nonlinear feedback control.
\begin{figure}[h!]
\centering
\psfrag{in_1}{$u$}
\psfrag{y_1}{$y$}
\psfrag{u_2}{}
\psfrag{y_2}{}
\psfrag{r}{\hspace{0.35cm}$v$}
\psfrag{plant}{$\mc H$}
\psfrag{con}{\hspace{0.1cm}$\phi(y)$}
\includegraphics[width=8.5cm]{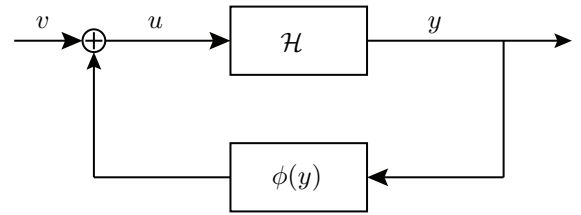}
\caption{Closed-loop interconnection of the system $\mc H$ given in (\ref{eq:general nonlinear system}) and a memoryless nonlinearity $\phi(y)$.}
\label{fig:interconnection}
\end{figure}

\subsection{NI preservation}
Consider a system $\mc H$ with the state-space model (\ref{eq:general nonlinear system}), whose input is given by
\begin{equation}\label{eq:control input}
	u = \phi(y)+v,
\end{equation}
where $\phi:\mb R^p\to \mb R^p$ is a time-invariant, memoryless nonlinear function with $\phi(0)=0$, and $v\in\mb R^p$ is a new input. The resulting system can be represented by the following equations.
\begin{subequations}\label{eq:CL with input v}
	\begin{align}
    \dot x=&\ f(x,\phi(h(x))+v),\\
    y=&\ h(x).
\end{align}
\end{subequations}

\begin{theorem}\label{theorem:NI preservation}
	Consider the system $\mc H$ with the state-space model (\ref{eq:general nonlinear system}), and  input (\ref{eq:control input}). Suppose the system $\mc H$ is NI with storage function $V(x)$ satisfying (\ref{eq:NI inequality}). Also, suppose there exists a continuously differentiable function $F:\mb R^p\to \mb R$ such that $\nabla F(y) = \phi(y)$ and $W(x):= V(x)-F(h(x))$ is positive definite. Then, the resulting closed-loop system (\ref{eq:CL with input v}) as shown in Fig.~\ref{fig:interconnection} is also NI with storage function $W(x)$.
\end{theorem}
\begin{IEEEproof}
	We show that the NI condition (\ref{eq:NI inequality}) is satisfied for the system (\ref{eq:CL with input v}) with storage function $V(x)-F(h(x))$ and the new input $v$ in the following.
	\begin{align}
		\dot V(x)-\dot F(h(x))-v^{\top}\dot y =&\ \dot V(x)-\phi(y)^{\top}\dot y-v^{\top}\dot y\notag\\
		 =&\ \dot V(x)-u^{\top}\dot y\leq 0.
	\end{align}
\end{IEEEproof}
\begin{remark}
	The result in Theorem \ref{theorem:NI preservation} shows that the storage function of an NI system can be modified via a suitable static nonlinear feedback control, in the part that depends on the system output. Here, $F(y)$ can be positive definite, negative definite, or sign indefinite, as long as the new storage $W(x)$ is positive definite.
\end{remark}

\subsection{Physical interpretation -- potential shaping}
A typical example of NI systems is a mechanical system with colocated force actuators and position sensors. For such systems, the model (\ref{eq:general nonlinear system}) can be written in a more detailed form:
\begin{subequations}
	\begin{align}
		\dot x_1 =&\ x_2,\\
		\dot x_2 =&\ f_2(x_1,x_2,u),\\
		y =&\ h_1(x_1),
	\end{align}
\end{subequations}
where $x_1,x_2\in \mb R^{\frac{n}{2}}$ ($\frac{n}{2}\in \mb N_+$) denotes the displacements and velocities, respectively. In this case, the output $y$ measures a full set or a subset of the displacements and $\dim(y)=p\leq \frac{n}{2}$.

The static nonlinearity $\phi(y)=\phi(h_1(x_1))$ reshapes the potential energy of the system. Specifically, suppose $V(x) = V_p(x_1)+V_k(x_2)$, where $V_p(x_1)$ and $V_k(x_2)$ are the potential and kinetic energy of the original system $\mc H$, respectively. Then the resulting system has the same kinetic energy and a modified potential energy $\widehat V_p(x_1)=V_p(x_1)-F(h_1(x_1))$, where $F$ satisfies the same conditions required in Theorem \ref{theorem:NI preservation} for this case; i.e., $\nabla F(y)=\phi(y)$ and $\widehat V_p(x_1)$ is positive definite.

For a mechanical plant $\mc H$, the feedback acts as if virtual springs are introduced between the measured coordinates. By appropriate choice of $F$, these springs can stiffen, soften, or couple joints in ways not present in the physical plant, while preserving the NI property. This type of potential shaping is directly related to compliant control, where feedback modifies the apparent stiffness of the system. Moreover, since the closed-loop remains NI, the shaped plant can still be treated within the NI control framework for subsequent design.

\section{Absolute stability for NI systems}\label{sec:absolute stability}
In this section, we consider the absolute stability problem for nonlinear NI systems. That is, we provide conditions under which the interconnection of the nonlinear NI system $\mc H$ and the static nonlinearity $\phi(y)$ as shown in Fig.~\ref{fig:interconnection} with $v=0$ is asymptotically stable. We first consider NI systems in a general nonlinear system structure. Then we specialize to input-affine nonlinear NI systems in an input-output Hamiltonian form. Existing absolute stability results for linear NI systems are recovered as special cases of the proposed framework.

\subsection{Absolute stability for NI systems in a general nonlinear system form}
In order to achieve closed-loop asymptotic stability, we require certain strictness in the plant. We consider nonlinear OSNI systems in this paper. Absolute stability for NI systems with other types of strictness can be analyzed in a similar manner.

We use the following observability-type assumption on the output map to exclude
hidden internal motion that is invisible from the output.
\begin{assumption}\label{assumption1}
For any trajectory of the system in the interval $[t_a,t_b]$, if the system output $y=h(x)$ remains constant, then the system state $x(t)$ also remains constant.
\end{assumption}

In this paper, absolute stability refers to asymptotic stability of the
closed-loop origin for the class of static nonlinearities satisfying the stated
conditions.

\begin{theorem}\label{thm:absolute stability}
Consider a nonlinear OSNI system $\mc H$ given by (\ref{eq:general nonlinear system}) with storage function $V(x)$. Also, consider a memoryless nonlinear feedback
\begin{equation}
	u = \phi(y), \qquad \phi(0)=0
\end{equation}
as shown in Fig.~\ref{fig:interconnection} with $v=0$. Suppose Assumption \ref{assumption1} is satisfied and $f(x,\phi(h(x)))\neq 0$ for all $x\neq 0$. If there exists a continuously differentiable function $F:\mb R^p\to \mb R$ such that $\nabla F(y) = \phi(y)$ and $W(x):=V(x)-F(h(x))$ is positive definite. Then the closed-loop system is absolutely stable with the Lyapunov function $W(x)$.
\end{theorem}
\begin{IEEEproof}
Taking the time derivative of the Lyapunov function yields
	\begin{align}
		\dot W(x)=&\ \dot V(x)-\dot F(h(x)) = \dot V(x)-\nabla F(y)^\top\dot y \notag\\
		\le &\ u^\top\dot y-\epsilon\|\dot y\|^2- u^\top\dot y
= -\epsilon\|\dot y\|^2 \le 0.
	\end{align}
Therefore, the closed-loop system is Lyapunov stable. Now we investigate the case where $\dot W(x)\equiv 0$. In this case, we have that $\dot y\equiv 0$, which implies $\dot x \equiv 0$ according to Assumption \ref{assumption1}. Therefore, we have $f(x,\phi(h(x)))\equiv 0$, which implies $x=0$. Hence, the closed-loop system is asymptotically stable.
\end{IEEEproof}

\subsection{Absolute stability for NI systems in an input-output Hamiltonian form}
Physical models with NI properties can usually be written into the input-output Hamiltonian form (\ref{eq:Hamiltonian}). In the following, we provide an absolute stability result for the system (\ref{eq:Hamiltonian}) with $R(x)$ assumed to be positive definite for all $x\in \mb R^n$. In this case, Assumption \ref{assumption1} is no longer required.

\begin{theorem}\label{thm:absolute stability Hamiltonian}
Consider the interconnection of the system $\mc H$ given by (\ref{eq:Hamiltonian}) with the memoryless nonlinearity
\begin{equation}
	u = \phi(y), \qquad \phi(0)=0
\end{equation}
as shown in Fig.~\ref{fig:interconnection} (with $v=0$). Suppose $H(x)$ is a positive definite function and $R(x)$ is a positive definite matrix for all $x\in \mb R^n$. Suppose there exists a continuously differentiable function $F:\mb R^p\to \mb R$ such that $\nabla F(y) = \phi(y)$ and $W(x):=H(x)-F(C(x))$ is positive definite with $\nabla W(x)\neq 0$ for all $x\neq 0$. Then the closed-loop system is absolutely stable with the Lyapunov function $W(x)$.
\end{theorem}
\begin{IEEEproof}
	We apply Lyapunov's stability theorem with candidate Lyapunov function $W(x)$. We have that
	\begin{align}
		\dot W(x)=&\ \dot H(x)-\dot F(C(x))\notag\\
		=&\ \nabla H(x)^{\top} \dot x-\nabla F(y) ^{\top}\nabla C(x)\dot x\notag\\
		=&\ \nabla H(x)^{\top}\dot x-\phi(y)^{\top}\nabla C(x)\dot x\notag\\
		=& \left(\nabla H(x)^{\top} - \phi(y)^{\top}\nabla C(x)\right)\dot x\notag\\
		=& \left(\nabla H(x)^{\top} - \phi(y)^{\top}\nabla C(x)\right)[J(x)-R(x)]\notag\\
		&\quad \left(\nabla H(x)-\nabla C(x)^{\top}\phi(y)\right)\notag\\
		=& -\left(\nabla H(x)^{\top} - \phi(y)^{\top}\nabla C(x)\right)R(x)\notag\\
		&\quad \left(\nabla H(x)-\nabla C(x)^{\top}\phi(y)\right)\notag\\
        =&\ -\nabla W(x)^\top R(x)\nabla W(x).\label{eq:dot lf}
	\end{align}
Because $R(x)>0$ and $\nabla W(x)\neq 0$ for all $x\neq 0$, it follows that $\dot W(x)<0$ for all $x\neq 0$. Hence the origin is asymptotically stable.
\end{IEEEproof}

\begin{remark}
    The potential energy shaping results in Section \ref{sec:preservation} and the absolute stability results in the present section naturally integrate with broader NI control architectures, including those incorporating dynamic feedback. To illustrate this viewpoint, consider the augmented configuration in Fig. \ref{fig:augmented interconnection}, which can be regarded as the interconnection of the Lur'e structure in Fig.~\ref{fig:interconnection} and an additional dynamic NI system $\mc H_2$.
    
    In such a setting, the static nonlinearity $\phi(y_1)$ shapes the potential energy of the plant $\mc H_1$, while a dynamic controller $\mc H_2$ may be used to inject damping. For example, the potential shaping results can be applied along with the learning-based NI control results in \cite{shi2025negative}. Instead of learning an NI neural ODE (NINODE) that stabilizes a plant as in \cite{shi2025negative}, one can consider learning both a static nonlinearity $\phi(y_1)$ to shape the resulting system's potential energy and a NINODE $\mc H_2$ to provide sufficient damping. Applications of such a result arise naturally in tasks such as impedance control.
    
This framework also provides a way to interpret situations where an input feedthrough term exists in the controller. Suppose both $\mc H_1$ and $\mc H_2$ have no input feedthrough terms, then the combination of $\mc H_2$ and $\phi(y_1)$ can be regarded as a system with input feedthrough. By applying the proposed results, the interconnection of $\mc H_1$ and $\phi(y_1)$ provides an augmented system of the form (\ref{eq:CL with input v}). Then, the system in Fig.~\ref{fig:augmented interconnection} becomes the interconnection of two NI systems without input feedthrough terms, for which stability analysis is usually more straightforward. Note that current nonlinear NI control literature can deal with input feedthrough (see e.g., \cite{shi2023output,chen2024nonlinear,chen2024stability}), but is limited to diagonal nonlinearities. The above-mentioned extensions will follow as a direct application of the proposed results and hence will not be discussed in detail, as they are not the main focus of the present article.
\end{remark}
\begin{figure}[h!]
\centering
\psfrag{in_1}{$u_1$}
\psfrag{y_1}{$y_1$}
\psfrag{u_2}{$u_2$}
\psfrag{y_2}{$y_2$}
\psfrag{r}{\hspace{0.35cm}$v$}
\psfrag{plant}{$\mc H_1$}
\psfrag{phi}{\hspace{0.1cm}$\phi(y_1)$}
\psfrag{con}{\hspace{0.1cm}$\mc H_2$}
\includegraphics[width=8.5cm]{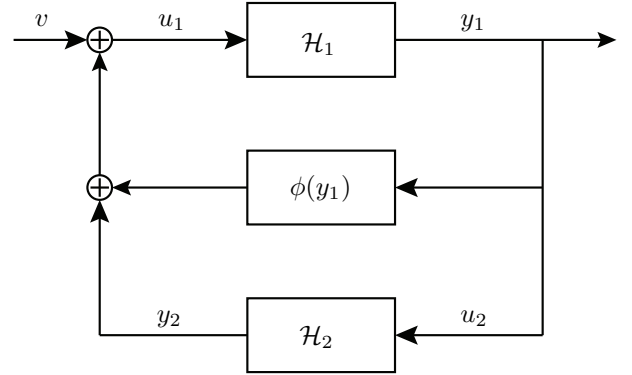}
\caption{Closed-loop interconnection of the system $\mc H_1$, a memoryless nonlinearity $\phi(y)$ and a system $\mc H_2$.}
\label{fig:augmented interconnection}
\end{figure}

\subsection{Comparison to the linear result}
The absolute stability problem is considered for linear NI systems in \cite{dey2016absolute}. In \cite{dey2016absolute}, conditions are provided under which the closed-loop interconnection of a strongly strictly negative imaginary (SSNI) system and a diagonal, memoryless, slope-restricted nonlinearity is absolutely stable. However, this is only a special case of our results even in the linear case as we allow for cross-coupling between channels in the static nonlinearity. We provide a more general linear result in this subsection. First, we present the results in \cite{dey2016absolute} for comparison.

Consider a linear system with a proper, real-rational, square transfer function $G(s)$ and the minimal realization $(A,B,C)$, whose input is given by a memoryless nonlinearity $\phi(y)$:
\begin{subequations}\label{eq:minimal linear system lure}
	\begin{align}
		\dot x =&\ Ax+Bu,\\
		y =&\ Cx,\\
		u =&\ \phi(y),
	\end{align}
\end{subequations}
where $x\in \mb R^n$, $u,y\in \mb R^p$ are the state, input and output. Here, $A\in \mb R^{n\times n}$, $B\in \mb R^{n\times p}$ and $C\in \mb R^{p\times n}$.

\begin{theorem}[{\cite[Thm.~9]{dey2016absolute}}]\label{thm:dey}
Consider the positive–feedback interconnection of a
minimal linear system $(A,B,C)$ with transfer function 
$G(s)$ and a diagonal, memoryless, 
slope-restricted nonlinearity $\phi(y)
   = [\phi_1(y_1),\cdots,\phi_p(y_p)]^\top$
satisfying
$0\le \phi_i'(y_i)<\mu_i$, $\phi_i(0)=0$, for all $i=1,\dots,p$.
Let $M=\mathrm{diag}(\mu_1,\ldots,\mu_p)>0$ denote the matrix of maximum slope bounds. Then the closed-loop system (\ref{eq:minimal linear system lure}) is absolutely stable if there exists a matrix $Y=Y^\top>0$ such that
\begin{equation}\label{eq:SSNI LMI}
   AY+YA^\top<0,\qquad B=-AYC^\top, 
\end{equation}
and the condition
\begin{equation}
   M^{-1}-G(0)>0
\end{equation}
holds.
\end{theorem}

We provide in the following the linear specialization of the proposed absolute stability results given in Theorem \ref{thm:absolute stability Hamiltonian}. Also, we consider a strongly strictly negative imaginary (SSNI) system, as is considered in Theorem~\ref{thm:dey} via \eqref{eq:SSNI LMI}.
\begin{theorem}\label{thm:absolute stability linear}
	Consider the positive feedback interconnection of minimal linear system $(A,B,C)$ and a memoryless nonlinearity $\phi(y)$ with $\phi(0)=0$. Suppose there exists a matrix $Y=Y^\top >0$ such that
	\begin{equation}
		AY+YA^{\top}<0, \text{ and } B = -AYC^{\top}.
	\end{equation}
Then the closed-loop system (\ref{eq:minimal linear system lure}) is absolutely stable if there exists a continuously differentiable function $F(y):\mb R^p\to \mb R$ such that $\nabla F(y) = \phi(y)$ and the function $W(x):=\frac{1}{2}x^{\top}Y^{-1}x-F(Cx)$ is positive definite with $\nabla W(x)\neq 0$ for all $x\neq 0$.
\end{theorem}
\begin{IEEEproof}
	We use the function $W(x)$ as a candidate Lyapunov function. We have that
	\begin{align}
		\dot W(x) =&\ x^{\top}Y^{-1}\dot x - \nabla F(Cx)^{\top}C\dot x\notag\\
		=&\ x^{\top}Y^{-1}\dot x - \phi(y)^{\top}C\dot x\notag\\
		=&\ \left(x^{\top}Y^{-1}-u^{\top}C\right)\dot x \notag\\
		=&\ \left(x^{\top}Y^{-1}-u^{\top}C\right)\left(Ax+Bu\right)\notag\\
		=&\ \left(x^{\top}Y^{-1}-u^{\top}C\right)\left(Ax-AYC^{\top}u\right)\notag\\
		=&\ \frac{1}{2}\left(x^{\top}Y^{-1}-u^{\top}C\right)(AY+YA^{\top})\left(Y^{-1}x-C^{\top}u\right)\notag\\
		\leq &\ 0.
		\end{align}
Since $AY+YA^\top<0$, we have $\dot W(x)<0$ whenever
$Y^{-1}x-C^\top u\neq 0$. Moreover,
\begin{equation}
	Y^{-1}x-C^{\top}u = Y^{-1}x-C^{\top}\phi(Cx)=\nabla W(x),
\end{equation}
which is nonzero for all $x\neq 0$ by assumption. Therefore, $\dot W(x)<0$ for all $x\neq 0$, and the origin is asymptotically stable.
\end{IEEEproof}

We show in the following that the conditions in Theorem \ref{thm:dey} are a special case of Theorem \ref{thm:absolute stability linear}.
\begin{lemma}
Consider a minimal linear system $(A,B,C)$ with $B=-AYC^\top$ where $Y=Y^\top > 0$. 
Let $\phi(y)=[\phi_1(y_1),\cdots,\phi_p(y_p)]^\top$ satisfy
$0\le \phi_i'(y_i)<\mu_i$ for all $i=1,\dots,p$.
Let $M=\mathrm{diag}(\mu_1,\ldots,\mu_p)>0$ denote the matrix of maximum slope bounds. Suppose $M^{-1}-G(0)>0$, then the function
$F(y)=\sum_i\int_0^{y_i}\phi_i(s)\,ds$ is such that $W(x)=\tfrac12 x^\top Y^{-1}x - F(Cx)$ is positive definite and $\nabla W(x)\neq 0$ for all $x\neq0$.
\end{lemma}
\begin{IEEEproof}
The slope bounds imply $F(y)\le\tfrac12 y^\top M y$.  
Moreover, since $G(0)=-CA^{-1}B=CYC^\top$, the condition
$M^{-1}-G(0)>0$ gives $M^{-1}-CYC^\top>0$. By the Schur complement, this is
equivalent to $Y^{-1}-C^\top M C>0$. Therefore, we have that
$W(x)\ge\tfrac12 x^\top(Y^{-1}-C^\top M C)x>0$ for all $x\neq 0$. It remains to show that $\nabla W(x)\neq 0$ for all $x\neq 0$. Suppose, for
contradiction, that $\nabla W(x)=0$. Then $Y^{-1}x=C^\top\phi(Cx)$. Let $y=Cx$. It follows that $y=CYC^\top\phi(y)=G(0)\phi(y)$. Multiplying both sides by $\phi(y)^\top$ gives $\phi(y)^\top y=\phi(y)^\top G(0)\phi(y)$. From $0\leq \phi_i'(s)<\mu_i$ and $\phi_i(0)=0$, we have $\phi_i(y_i)y_i\geq \mu_i^{-1}\phi_i(y_i)^2$ for each $i$. Hence $\phi(y)^\top y\geq \phi(y)^\top M^{-1}\phi(y)$. Since $M^{-1}-G(0)>0$, the equality above is impossible unless $\phi(y)=0$. Thus $y=0$, and $\nabla W(x)=0$ then implies $x=0$. Therefore, $\nabla W(x)\neq 0$ for all $x\neq 0$.
\end{IEEEproof}

\begin{remark}
In mechanical systems with multiple measured coordinates, the diagonal nonlinearity considered in Theorem~\ref{thm:dey} corresponds to virtually adjusting springs only between each measured coordinate and the fixed ground. Such a restriction does not allow for virtual coupling between two measured coordinates, thereby excluding many useful interconnection patterns. Moreover, the requirement of a positive gradient ($0\le\phi_i'(y_i)$) limits the feedback to ``softening'' the existing springs, while neglecting other physically meaningful cases such as stiffening or reshaping the potential field. In contrast, our results admit coupled and sign-indefinite potential change $F(y)$, allowing virtual springs between any measured coordinates and enabling more flexible energy shaping, which is important in applications that require coordinated or compliant multi-point interaction.

A subsequent comment paper \cite{carrasco2017comment} generalized \cite{dey2016absolute} by relaxing the slope-restricted condition on the nonlinearity to a sector-bounded condition. It also discussed the case of negative feedback, which corresponds to the situation where $F(y)$ is negative definite in our formulation. Nevertheless, the nonlinearities in \cite{carrasco2017comment} remain diagonal, and sign-indefinite $F(y)$ were not considered. Hence, both \cite{dey2016absolute,carrasco2017comment} represent special cases of the results developed in this paper, which offer broader applicability even within the linear setting.
\end{remark}

\section{Examples}\label{sec:examples}
In this section, we provide two examples to illustrate the proposed results. In Section \ref{sec:example linear}, we consider a linear example, for comparison with results in \cite{dey2016absolute} and \cite{carrasco2017comment}. In Section \ref{sec:example nonlinear}, we consider a nonlinear example and use the proposed results to shape its potential energy.
\subsection{Linear example -- comparison with existing results}\label{sec:example linear}
In this example, we consider the interconnection of a linear SSNI system and two types of memoryless nonlinearities that cannot be analyzed using the results in \cite{dey2016absolute} and \cite{carrasco2017comment}. However, these two cases can be dealt with using the results in the present paper.

Consider the minimal linear system \eqref{eq:minimal linear system lure} with
\begin{equation}\label{eq:example linear ABC}
	A=\begin{bmatrix}-1&0\\0&-2\end{bmatrix},\quad B=\begin{bmatrix}1&0\\0&2\end{bmatrix}, \quad 
C=I_2.
\end{equation}
With $Y=I_2$, we have $AY+YA^\top=2A<0$ and $B=-AYC^\top$. Consider a memoryless nonlinear feedback $u = \phi(y)$, where two different functions $\phi$ are considered in the following.

\paragraph*{(a) Diagonal nonlinearity, with sign indefinite $F(y)$}

Let $\phi_1(y) = \begin{bmatrix}
	0.2 y_1 & -0.5y_2
\end{bmatrix}^\top $. Such a $\phi_1$ does not satisfy the slope-restriction in \cite{dey2016absolute} or the sector-boundedness in \cite{carrasco2017comment}. However, $\phi_1(y)$ satisfies the conditions in Theorem \ref{thm:absolute stability linear} with the function $F_1(y) = 0.1 y_1^2-0.25 y_2^2$. According to Theorem \ref{thm:absolute stability linear}, we construct the candidate Lyapunov function as
$$W_1(x)=0.4x_1^2+0.75x_2^2,$$
which is positive definite. Hence the closed-loop system is asymptotically stable by Theorem~\ref{thm:absolute stability linear}. The asymptotic stability can be verified by observing the closed-loop state equation $\dot x = A_{cl}x$, where $A_{cl} = \begin{bmatrix}
	-0.8&0\\0&-3
\end{bmatrix}$ is Hurwitz.

\paragraph*{(b) Coupled nonlinearity}
Let
\begin{equation}\label{eq:example linear phi_2}
	\phi_2(y) = \begin{bmatrix}
	\sin(y_2-y_1)& \sin(y_1-y_2)
\end{bmatrix}^\top.
\end{equation}
Such a $\phi_2$ has cross-coupling between channels, and is not covered by \cite{dey2016absolute} and \cite{carrasco2017comment}. However, $\phi_2(y)$ satisfies the conditions in Theorem \ref{thm:absolute stability linear} with the function $F_2(y)=\cos(y_1-y_2)-1$. According to Theorem \ref{thm:absolute stability linear}, we construct the candidate Lyapunov function as
$$W_2(x) = 0.5x_1^2+0.5x_2^2+1-\cos(x_1-x_2),$$
which is positive definite. Hence, the closed-loop system is asymptotically stable by Theorem \ref{thm:absolute stability linear}. Fig.~\ref{fig:example linear} shows the state trajectories of the closed-loop system with initial state $\begin{bmatrix}
	1&-2
\end{bmatrix}^\top$.

\begin{figure}[h!]
\centering
\psfrag{State trajectories}{\hspace{-0.3cm}\textbf{State Trajectories}}
\psfrag{State}{State}
\psfrag{Time (s)}{Time ($\mathrm{s}$)}
\psfrag{x1}{\scriptsize$x_1$}
\psfrag{x2}{\scriptsize$x_2$}
\psfrag{1}{\scriptsize$1$}
\psfrag{0}{\scriptsize$0$}
\psfrag{-1}{\hspace{-1.23mm}\scriptsize$-1$}
\psfrag{-2}{\hspace{-1.23mm}\scriptsize$-2$}
\psfrag{0.5}{\scriptsize$0.5$}
\psfrag{1.5}{\scriptsize$1.5$}
\psfrag{2}{\scriptsize$2$}
\psfrag{2.5}{\scriptsize$2.5$}
\psfrag{3}{\scriptsize$3$}
\includegraphics[width=8.5cm]{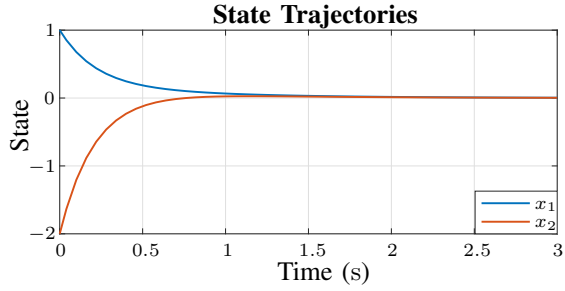}
\caption{State trajectories of the system \eqref{eq:minimal linear system lure} with (\ref{eq:example linear ABC}) and (\ref{eq:example linear phi_2}).}
\label{fig:example linear}
\end{figure}
\subsection{Nonlinear example -- potential energy shaping}\label{sec:example nonlinear}
Consider a mechanical plant consisting of two pendulums hinged at a common pivot, as illustrated in Fig.~\ref{fig:two_pendulums}. The pendulums swing in adjacent, parallel planes, separated sufficiently to prevent contact. Each pendulum \(i \in \{1,2\}\) has mass \(m_i\), rod length \(\ell_i\), angular displacement \(\theta_i\), and angular velocity \(\omega_i\). Each pendulum is influenced by gravity, a torsional spring with stiffness \(k_i\), a viscous damper with coefficient \(d_i\), and an external control torque \(u_i\). The two pendulums are additionally coupled through a torsional spring \(k_c\) and a damper \(d_c\) acting at the common pivot. The system parameters are given as follows, with all quantities expressed in standard SI units:
\begin{align}
&m_1=2, m_2=1.5, k_1=2, k_2=1,  \ell_1=1, \ell_2=1, d_1=0.5,\notag\\
&d_2=0.8, k_c=0.2, d_c=1.5, g\approx 9.81.
\end{align}
\begin{figure}[h!]
\centering
\begin{tikzpicture}
    \pgfmathsetmacro{\Gvec}{1}
    \pgfmathsetmacro{\angOne}{15}     
    \pgfmathsetmacro{\angTwo}{60}    
    \pgfmathsetmacro{\Lone}{3.0}      
    \pgfmathsetmacro{\Ltwo}{3.0}      

    \coordinate (O) at (0,0);
    \draw[dashed,gray,line width=0.6mm] (O) -- ++(0,-4) coordinate (ref);
    \filldraw[fill=black!60] (O) circle (0.1);

    \coordinate (bob1) at ($(O)+(270+\angOne:\Lone)$);
    \draw[line width=0.5mm] (O) -- (bob1);
    \filldraw[fill=red!50,draw=black] (bob1) circle (0.15);
    \pic [draw,-stealth,"$\theta_1$",angle eccentricity=1.2,angle radius=1.5cm,line width=0.2mm]
        {angle = ref--O--bob1};
    \node at ($(bob1)+(0.5,-0.2)$) {$m_1$};

    \coordinate (bob2) at ($(O)+(270+\angTwo:\Ltwo)$);
    \draw[line width=0.5mm] (O) -- (bob2);
    \filldraw[fill=red!50,draw=black] (bob2) circle (0.15);
    \pic [draw,-stealth,"$\theta_2$",angle eccentricity=1.3,angle radius=0.9cm,line width=0.2mm]
        {angle = ref--O--bob2};
    \node at ($(bob2)+(0.5,-0.2)$) {$m_2$};
\end{tikzpicture}
\caption{Two pendulums with hinge springs $k_1,k_2$, hinge dampers $d_1,d_2$, coupling spring $k_c$, coupling damper $d_c$, gravity $g$, and input torques $u_1,u_2$.}
\label{fig:two_pendulums}
\end{figure}
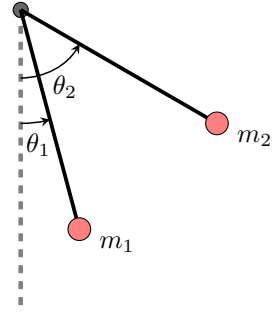

The state-space model of the plant is given by
\begin{subequations}
	\begin{align}
		\dot \theta_1 =&\ \omega_1,\\
		\dot \theta_2 =&\ \omega_2,\\
		\dot \omega_1 =&\ \frac{1}{m_1 \ell_1^2}\Big(-m_1g\ell_1\sin\theta_1-k_1\theta_1-d_1\omega_1\notag\\
		&-k_c(\theta_1-\theta_2)-d_c(\omega_1-\omega_2)+u_1\Big),\\
		\dot \omega_2 =&\ \frac{1}{m_2 \ell_2^2}\Big(-m_2g\ell_2\sin\theta_2-k_2\theta_2-d_2\omega_2\notag\\
		&+k_c(\theta_1-\theta_2)+d_c(\omega_1-\omega_2)+u_2\Big),\\
		y=&\ \begin{bmatrix}
			\theta_1 & \theta_2
		\end{bmatrix}^\top.
	\end{align}
\end{subequations}
Such a pendulum system is OSNI as can be verified using Definition \ref{def:nonlinear OSNI} with storage function
\begin{align}
	V(&\theta_1,\theta_2,\omega_1,\omega_2) = \frac{1}{2}k_c(\theta_1-\theta_2)^2\notag\\
	&+\sum_i\Big(\frac{1}{2}k_i\theta_i^2+\frac{1}{2} m_i\ell_i^2 \omega_i^2+m_ig\ell_i(1-\cos\theta_i)\Big).
\end{align}
The equilibrium at the origin is asymptotically stable. However, because of the gravitational forces, the region of attraction is not the entire state space. Consequently, the pendulums may settle at other nonzero equilibria unless initialized close to the origin. In the following, we use this model to illustrate how we apply the proposed results to shape the potential energy of the system in order to
\begin{enumerate}
\item stiffen the coupling between the pendulums so that $\theta_1$ and $\theta_2$ remain nearly synchronized during oscillatory motion under new external input;
\item enlarge the region of attraction to be the full state space so that the $\theta_1$ and $\theta_2$ converge to the origin rather than other nonzero equilibria.
\end{enumerate}
To enforce synchronization between the two pendulums, we introduce additional virtual potential–energy terms that penalize the relative displacement $e := \theta_1 - \theta_2$.  
Let the new total storage function be
\begin{equation}
	V_1 = V + \beta e^2 + \kappa\big(\sqrt{e^2+\delta^2}-\delta\big),
\end{equation}
where $\beta = 1.5$, $\kappa = 5$, and $\delta = 0.1$.  

The quadratic term $\beta e^2$ acts as a virtual linear torsional spring coupling the two pendulums. However, the coupling effect is weak for small $|e|$. Therefore, we also include the nonlinear term $\kappa(\sqrt{e^2+\delta^2}-\delta)$ which reshapes the potential energy to be steep near $e=0$ and gradually flatter for large $|e|$. Together, these two terms provide effective coupling without excessive stiffness. According to Theorem \ref{theorem:NI preservation}, to achieve the storage function $V_1$, we apply the control input $u = \phi_1(y)+v$ where
\begin{equation}
	\phi_1(y) = \nabla F_1(y) = \nabla_y (V-V_1) = \begin{bmatrix}
		-2\beta e-\kappa \frac{e}{\sqrt{e^2+\delta^2}}\\2\beta e+\kappa \frac{e}{\sqrt{e^2+\delta^2}}
	\end{bmatrix}.
\end{equation}
To illustrate the effect of the enhanced coupling, we simulate the trajectories of $\theta_1$ and $\theta_2$ for both the original and the shaped systems.  
To highlight the synchronization behavior, for both the original and shaped system, a square–wave torque is applied only to pendulum $m_1$, with amplitude $2$ and period $3~\mathrm{s}$; that is,
\[
v_1(t) = 2\,\mathrm{sgn}\!\left(\sin\!\left(\tfrac{2\pi}{3}t\right)\right), \qquad v_2(t)=0.
\]
For both systems, we set the initial angular displacements to be
\begin{equation}\label{eq:example initial conditions}
	\begin{bmatrix}\theta_1 & \theta_2\end{bmatrix}^\top = \begin{bmatrix}6 & 4.5\end{bmatrix}^\top.
\end{equation}  
As shown in Fig.~\ref{fig:example nonlinear enhanced coupling}, the pendulums in the original system exhibit noticeable phase and amplitude differences, whereas the shaped system achieves rapid synchronization, where $\theta_1-\theta_2$ quickly decreases and remains small while both pendulums continue to oscillate under the external forcing.

\begin{figure}[h!]
\centering
\psfrag{theta1 and theta2 and their differences}{\hspace{0.4cm}\textbf{Trajectories of $\theta_1$, $\theta_2$ and $\theta_1-\theta_2$}}
\psfrag{Original}{Original}
\psfrag{Enhanced Coupling}{Enhanced Coupling}
\psfrag{Time}{Time ($\mathrm{s}$)}
\psfrag{th1}{$\theta_1$}
\psfrag{th2}{$\theta_2$}
\psfrag{difference}{$\theta_1-\theta_2$}
\psfrag{x2}{\scriptsize$x_2$}
\psfrag{-2}{\hspace{-1.23mm}\scriptsize$-2$}
\psfrag{0}{\scriptsize$0$}
\psfrag{2}{\scriptsize$2$}
\psfrag{4}{\scriptsize$4$}
\psfrag{6}{\scriptsize$6$}
\psfrag{8}{\scriptsize$8$}
\psfrag{5}{\scriptsize$5$}
\psfrag{10}{\scriptsize$10$}
\psfrag{15}{\scriptsize$15$}
\psfrag{20}{\scriptsize$20$}
\psfrag{25}{\scriptsize$25$}
\psfrag{30}{\scriptsize$30$}
\includegraphics[width=9.5cm]{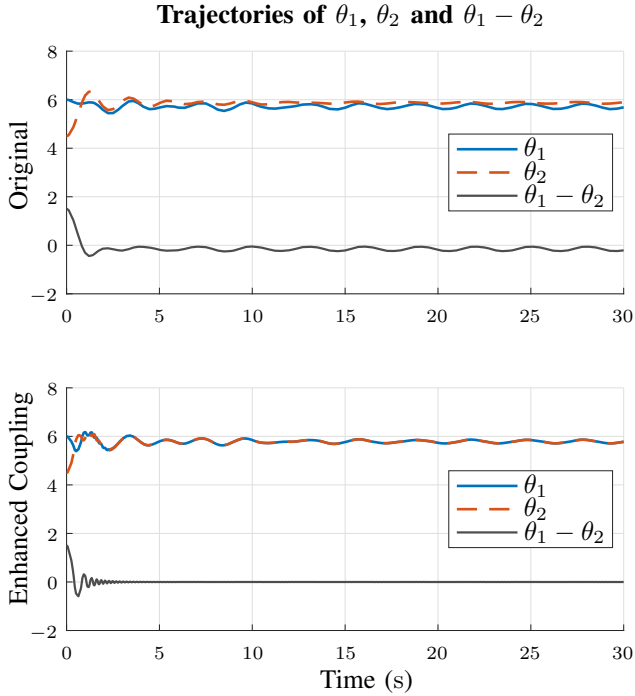}
\caption{Angular displacements of the two pendulums and their relative difference for the original plant and the shaped plant with enhanced coupling. A square-wave torque is applied to pendulum $m_1$ to illustrate the synchronization effect.}
\label{fig:example nonlinear enhanced coupling}
\end{figure}

It should be noted that although $\theta_1$ and $\theta_2$ can be synchronized, they oscillate around a nonzero angle rather than the origin. This occurs because the region of attraction does not cover the entire state space, and multiple local minima exist in the potential energy. To enlarge the region of attraction and ensure convergence to the unique equilibrium at the origin, we further reshape the potential energy by adding a term that eliminates the undesired local minima in the storage function. Specifically, let
\begin{equation}\label{eq:V2}
	V_2 = V_1 + \sum_i a\,\log(\cosh(b\,\theta_i)),
\end{equation}
where we choose the parameters $a=5$ and $b=3$. Here, $\sum_i a\,\log(\cosh(b\,\theta_i))$ is a positive-definite function that steepens the potential near the origin. This shaping term effectively flattens the distant wells created by the gravitational terms and yields a single-well potential energy surface with a global minimum at the origin. This can also be observed in Fig.~\ref{fig:potential shaping}, which compares the potential energy surface of the original system and the shaped system.

To obtain the shaped potential energy $V_2$ as given in (\ref{eq:V2}), we can apply the input
\begin{align}
\phi_2(y) =&\ \nabla F_2(y) = \nabla_y (V-V_2)\notag\\
=& \begin{bmatrix}
		-2\beta e-\kappa \frac{e}{\sqrt{e^2+\delta^2}}-ab\tanh(b\theta_1)\\2\beta e+\kappa \frac{e}{\sqrt{e^2+\delta^2}}-ab\tanh(b\theta_2)
	\end{bmatrix}.\label{eq:phi2}
\end{align}

\begin{figure}[h!]
\centering
\psfrag{Potential energy surfaces}{\textbf{Potential Energy Surfaces}}
\psfrag{Original potential}{\hspace{-10mm}Original potential $V(\theta_1,\theta_2,0,0)$}
\psfrag{Shaped potential}{\hspace{-10mm}Shaped potential $V_2(\theta_1,\theta_2,0,0)$}
\psfrag{theta1}{\hspace{3mm}$\theta_1$}
\psfrag{theta2}{\hspace{3mm}$\theta_2$}
\includegraphics[width=9.5cm]{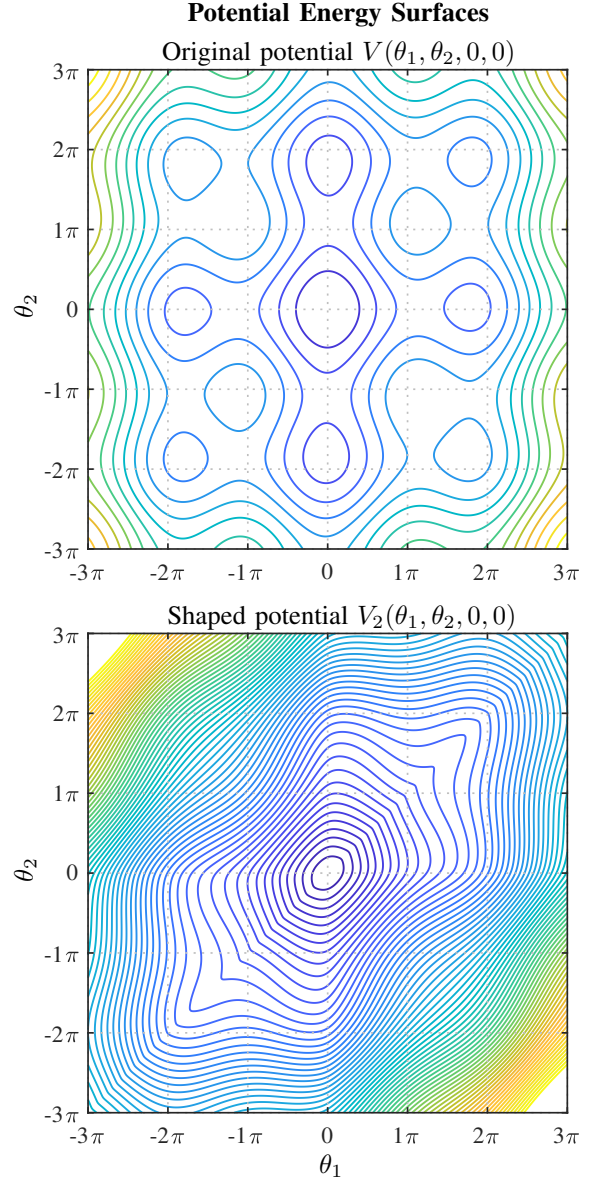}
\caption{Potential energy surfaces of the system before and after potential shaping. The original potential exhibits multiple local minima due to gravitational and coupling effects, while the shaped potential has only a minimum at the origin.}
\label{fig:potential shaping}
\end{figure}

Now we simulate the motion of the system for the plant with the same initial conditions (\ref{eq:example initial conditions}) and under $u=\phi_2(y)$. It can be observed that unlike the case in Fig.~\ref{fig:example nonlinear enhanced coupling}, the angular displacements $\theta_1$ and $\theta_2$ now converge to zero. Moreover, Fig.~\ref{fig:state trajectories} shows asymptotic stability of the system, which aligns with the results in Theorem \ref{thm:absolute stability}.

\begin{figure}[h!]
\centering
\psfrag{State Trajectories}{\textbf{State Trajectories}}
\psfrag{State}{State}
\psfrag{Time}{Time ($\mathrm{s}$)}
\psfrag{th1}{\small$\theta_1$}
\psfrag{th2}{\small$\theta_2$}
\psfrag{om1}{\small$\omega_1$}
\psfrag{om2}{\small$\omega_2$}
\psfrag{10}{\scriptsize$10$}
\psfrag{5}{\scriptsize$5$}
\psfrag{0}{\scriptsize$0$}
\psfrag{-5}{\hspace{-1.3mm}\scriptsize$-5$}
\psfrag{-10}{\hspace{-1.2mm}\scriptsize$-10$}
\psfrag{-15}{\hspace{-1.25mm}\scriptsize$-15$}
\psfrag{20}{\scriptsize$20$}
\psfrag{30}{\scriptsize$30$}
\psfrag{40}{\scriptsize$40$}
\psfrag{50}{\scriptsize$50$}
\includegraphics[width=9.5cm]{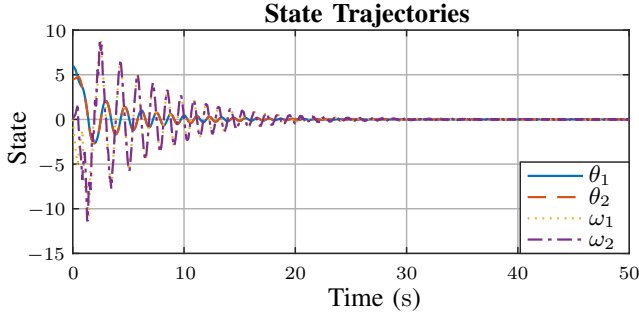}
\caption{State trajectories of the plant under input $u=\phi_2(y)$ with $\phi_2$ given by (\ref{eq:phi2}).}
\label{fig:state trajectories}
\end{figure}

\section{Conclusion}\label{sec:conclusion}
We provide conditions under which the interconnection of a nonlinear NI system and a static nonlinearity remains NI. Also, we provide an interpretation of the result from the perspective of storage function shaping. Based on this result, we derived absolute stability conditions for nonlinear NI systems with some level of strictness. Also, we show that the linear specialization of the results is strictly more general than existing absolute stability results for linear NI systems. This is then illustrated using a linear example. A nonlinear pendulum example is also studied to illustrate the utility of the proposed results in potential energy shaping.
\bibliographystyle{IEEEtran}
\bibliography{reference.bib}

\end{document}